\documentclass[submission]{eptcs}
\usepackage[utf8]{inputenc}

\usepackage{amsthm}
\usepackage{amssymb}
\usepackage{amsmath}
\usepackage{tikz}
\usepackage{rotating}
\usepackage{multirow}
\usepackage{textcomp}
\usepackage{MnSymbol,wasysym}
\usepackage{paralist}
\usepackage{comment}
\usepackage{textcomp}
\usepackage{stmaryrd}

\newif\ifdraft\draftfalse
\newcommand{\nullproc}{{\bf 0}}

\newcommand{\dedrule}[2]{\frac{#1}{#2}}
\newcommand{\trans}[1][]{\xrightarrow{\, {#1} \, }}
\newcommand{\transs}[1][]{\overset{\, #1\, }{\Longrightarrow}}
\newcommand{\ntrans}[1][]{\mathrel{{\trans[#1]}\hspace{-0.7em}/}{\!}}

\newcommand{\Nat}{\mathbb{N}_0}
\newcommand{\relR}{\mathrel{R}}

\newcommand{\bisim}{\sim}
\newcommand{\nbisim}{\not \sim}

\newcommand{\Var}{\mathcal{V}}

\newcommand{\closedT}[1][]{T}
\newcommand{\closedTerms}[1][]{T(\Sigma)}

\newcommand{\ov}[1]{\overline{#1}}

\newcommand{\ms}[1]{\ensuremath{\lbag{#1}\rbag}}

\newcommand{\porder}{\preccurlyeq}
\newcommand{\porderstrict}{\prec}
\newcommand{\fn}{\textit{fn}}
\newcommand{\bn}{\textit{bn}}

\newcommand{\pp}{{\bf P}}

\newcommand{\length}{\textit{length}}

\newcommand{\stops}{\downarrow}

\newcommand{\depth}[1]{\textit{depth}(#1)}
\newcommand{\norm}[1]{\textit{norm}(#1)}
\newcommand{\normt}{\textit{norm}'}

\newcommand{\finitePi}{\Pi_f}
\newcommand{\wossPi}{\Pi_{\scriptsize \not \circlearrowright}}

\newcommand{\wbisim}{\approx}
\newcommand{\nwbisim}{\not \approx}

\newcommand{\wec}[1]{ {[#1]}_{\wbisim} }
\newcommand{\wecRest}[1]{\wec{ #1 }^{\scriptsize \not \circlearrowright}}
\renewcommand{\sec}[1]{ {[#1]}_{\bisim} }

\newcommand{\wpp}{\pp_{\wbisim}}
\newcommand{\spp}{\pp_{\bisim}}

\newcommand{\parpi}{\mid}

\newcommand{\newtext}[1]{#1}

\usepackage[ddmmyyyy]{datetime}
\usepackage{datenumber}
\usepackage{expl3}
\usepackage{color}

\usepackage{etex}
\reserveinserts{30}

\makeatletter
\newcommand*\getnumtz[2]{%
    \expandafter\@getnumtz\the\numexpr 0#2\relax
        \empty\relax\relax\@nnil{#1}{#2}%
}

\def\@getnumtz#1\relax#2\relax#3\@nnil#4#5{%
    \ifx\relax#2\relax
        \edef#4{#1}%
    \else
        \begingroup\expandafter\endgroup
        \expandafter\let\expandafter#4\csname getnumtz@#5\endcsname%
    \fi
}

\newcommand*\definetz[2]{%
    \@namedef{getnumtz@#1}{#2}%
}%

\definetz{Z}{+0000}
\definetz{GMT}{+0000}
\definetz{UTC}{+0000}
\definetz{CET}{+0100}
\definetz{CEST}{+0200}

\newcommand*\converttimezone[9]{%
    \begingroup
    \c@myyear=\numexpr#2\relax
    \c@mymonth=\numexpr#3\relax
    \c@myday=\numexpr#4\relax
    \c@myhour=\numexpr#5\relax
    \c@myminute=\numexpr#6\relax
    \c@mysecond=\numexpr#7\relax
    \getnumtz\origtz{#8}%
    \getnumtz\targettz{#9}%
    \c@myhourminute=\numexpr (#5)*100+(#6) - \origtz + \targettz \relax
    \c@myhour=\numexpr \c@myhourminute / 100\relax% integer devision
    \c@myminute=\numexpr \c@myhourminute - \c@myhour*100\relax
    \loop\ifnum\c@myminute<\z@
        \advance\c@myhour by \m@ne
        \advance\c@myminute by 60\relax
    \repeat
    \loop\ifnum\c@myminute>59\relax
        \advance\c@myhour by \@ne
        \advance\c@myminute by -60\relax
    \repeat
    \ifnum\c@myhour<0\relax
        \setmydatenumber{mydatenumber}{\value{myyear}}{\value{mymonth}}{\value{myday}}%
        \advance\c@mydatenumber by \m@ne
        \setmydatebynumber{\value{mydatenumber}}{myyear}{mymonth}{myday}%
        \advance\c@myhour by 24\relax
    \else\ifnum\c@myhour>23\relax
        \setmydatenumber{mydatenumber}{\value{myyear}}{\value{mymonth}}{\value{myday}}%
        \advance\c@mydatenumber by \@ne
        \setmydatebynumber{\value{mydatenumber}}{myyear}{mymonth}{myday}%
        \advance\c@myhour by -24\relax
    \fi\fi
    \edef\@tempa{\unexpanded{#1}{\themyyear}{\themymonth}{\themyday}{\themyhour}{\themyminute}{\themysecond}{#9}}%
    \expandafter
    \endgroup\@tempa
}

\newcounter{myyear}
\newcounter{mymonth}
\newcounter{myday}
\newcounter{myhour}
\newcounter{myminute}
\newcounter{mysecond}
\newcounter{mydatenumber}
\makeatother

\ExplSyntaxOn
\cs_set_eq:NN \lgt_file_date: \today

\cs_new_nopar:Npn \lgt_file_time:
  {
    \lgt_get_time:N \currenthour :
    \lgt_get_time:N \currentminute :
    \lgt_get_time:N \currentsecond
  }

\cs_new_nopar:Npn \lgt_get_time:N #1
  {
    \int_compare:nNnT { #1 } < { 10 } { 0 }
    \int_use:N #1
  }

\ExplSyntaxOff

\definecolor{lightblue}{RGB}{224,224,255}
\definecolor{lightred}{RGB}{255,224,224}
\definecolor{lightgreen}{RGB}{224,255,224}
\definecolor{lightyellow}{RGB}{255,255,224}
\definecolor{lightpurple}{RGB}{255,224,255}
\definecolor{darkerred}{RGB}{64,0,0}
\definecolor{darkred}{RGB}{128,0,0}
\definecolor{darkblue}{RGB}{0,0,128}
\definecolor{darkgreen}{RGB}{0,128,0}
\definecolor{darkpurple}{RGB}{128,0,128}
\definecolor{grey}{rgb}{0.745098,0.745098,0.745098}
\definecolor{lightgrey}{rgb}{0.9,0.9,0.9}
\definecolor{darkgrey}{rgb}{0.6,0.6,0.6}

\newcommand{\colorpar}[3]{\colorbox{#1}{\parbox{#2}{#3}}}

\newcommand{\marginremark}[3]{\marginpar{\colorpar{#2}{0.65\linewidth}{\color{#1}#3}}}

\makeatletter
\def\THICKhrulefill{\leavevmode \leaders \hrule height 5pt\hfill \kern \z@}
\makeatother

\newcommand{\remarkBL}[1]{\marginremark{darkred}{lightred}{\tiny{[BL]~ #1}}}
\newcommand{\remarkMDL}[1]{\marginremark{darkgreen}{lightgreen}{\tiny{[MDL]~ #1}}}
\newcommand{\remarkPRD}[1]{\marginremark{darkpurple}{lightpurple}{\tiny{[PRD]~ #1}}}
\newcommand{\remarkWF}[1]{\marginremark{black}{lightyellow}{\tiny{[WF]~ #1}}}

\ifdraft
\else
\renewcommand{\remarkBL}[1]{}
\renewcommand{\remarkMDL}[1]{}
\renewcommand{\remarkPRD}[1]{}
\renewcommand{\remarkWF}[1]{}
\fi

\newtheorem{definition}{Definition}
\newtheorem{lemma}{Lemma}
\newtheorem{theorem}{Theorem}

\newtheorem{convention}{Convention}
\newtheorem{corollary}{Corollary}
\title{Unique Parallel Decomposition for the $\pi$-calculus
  \thanks{
    M.D. Lee has been supported by the project ANR 12IS02001 PACE}
}
\author{
Matias David Lee 
\institute{
  Univ. Lyon, ENS de Lyon, CNRS, UCB Lyon 1, LIP, France.
\email{matias-david.lee@ens-lyon.fr}}
\and
Bas Luttik
\institute{
 Eindhoven University of Technology, The Netherlands.
\email{s.p.luttik@tue.nl}}
}

\providecommand{\event}{EXPRESS/SOS 2016}
\providecommand{\firstpage}{1}
\providecommand{\titlerunning}{Unique Parallel Decomposition for the $\pi$-calculus}
\def\authorrunning{M.D. Lee and B. Luttik}

\begin{document}

\maketitle

\begin{abstract}
A (fragment of a) process algebra satisfies unique parallel decomposition if 
the definable behaviours admit a unique decomposition into indecomposable parallel components. 
In this paper we prove that finite processes of the $\pi$-calculus, i.e. processes that perform no infinite executions, 
satisfy this property modulo strong bisimilarity and weak bisimilarity.
Our results are obtained by an application of a general technique for
establishing unique parallel decomposition using decomposition orders.
\end{abstract}

\section{Introduction}

A (fragment of a) process algebra has \emph{unique parallel decomposition} (UPD) if 
all definable behaviours admit a unique decomposition into indecomposable parallel components. 
In this paper we prove that finite processes definable in the
$\pi$-calculus satisfy this property modulo strong bisimilarity and modulo weak bisimilarity.

From a theoretical point of view, this property is interesting because it can be used to
prove other theoretical properties about process calculi. 
For instance, relying on unique parallel decomposition,  Moller proves in \cite{Moller90b,Moller90a} that PA and CCS cannot be
finitely axiomatized without auxiliary operations, and Hirshfeld and Jerrum prove in \cite{HJ99}
that bisimilarity is decidable for normed PA.
Unique parallel decomposition can be also used to define a notion of normal form.
Such a notion of normal form is useful in completeness proofs for
equational axiomatizations in settings in which an elimination theorem
for parallel composition is lacking (see, e.g., \cite{AFIL05,AFIL09,AILT08,FL00,HP08}).
In \cite{LPSS11}, UPD is used to prove complete axiomatisation and decidability 
results in the context of a higher-order process calculus.

From a practical point of view, unique parallel decomposition can be
used to devise methods for finding the maximally parallel
implementation of a behaviour \cite{CGM98}, 
or for improving verification methods \cite{GM92}. 
In \cite{DLL12}, a unique parallel decomposition result is used as a tool in the comparison of different
security notions in the context of electronic voting.

% 
% There is an intimate relationship between unique parallel decomposition and
% of cancellation with respect to parallel composition; the properties are in most
% circumstances equivalent. In [6], cancellation with respect parallel composition
% was first proved and exploited to prove the completeness of an axiomatisation of
% distributed bisimilarity.

The UPD property has been widely studied for different process calculi
and variants of the parallel operator.
Milner and Moller were the first to establish a unique parallel
decomposition theorem; they proved the property for a simple process
calculus that allows the specification of all finite behaviours up to
strong bisimilarity and includes parallel composition in the form of
pure interleaving without interaction between its components \cite{MM93}. 
Moller, in his dissertation \cite{Moller89},  extended this result replacing interleaving parallel composition by CCS parallel composition,
and then also considering weak bisimilarity. 
% 
% These results were established with
% subsequent refinements of an ingenious proof technique attributed to Milner. 
Christensen, also in his dissertation \cite{Chris93}, proved unique
decomposition for normed behaviours recursively definable modulo strong bisimilarity, and for
all behaviours recursively definable modulo distributed bisimilarity;
the proof of the latter result relies on a cancellation law for
parallel composition up to distributed bisimilarity, first established
by Castellani as \cite[Lemma 4.14]{Cas88}.

Most of the aforementioned unique parallel decomposition results were established with subsequent refinements of an ingenious proof technique attributed to Milner.
% With each successive refinement of Milner’s proof technique, the technical details became more complicated, 
%but the general idea of the proof remained the same. 
% 
In \cite{LO05}, the notion of \emph{decomposition order} is introduced
in order to formulate a sufficient condition on
commutative monoids that facilitates an abstract version of Milner's
proof technique. It is then proved that if a partial commutative monoid can be
endowed with a decomposition order, then it has unique
decomposition. Thus, an algebraic tool is obtained that allows one to
prove UPD for a process calculus by finding a decomposition order. 

The tool can deal with most of the settings aforementioned.
%except for Moller's result that states that finite behaviours modulo weak bisimilarity have unique %decomposition.
% A decomposition order concisely state the sufficient condition, we
% have put forward the notion of decomposition order ;
In this paper, we show how the tool can also be applied to obtain
unique parallel decomposition results for finite processes of the $\pi$-calculus
w.r.t.\ strong bisimilarity and w.r.t.\ weak bisimilarity.
But, to this end, we do face two complications:
The first complication, in the context of the $\pi$-calculus is that,
as opposed to previous settings, the decomposition order is not
directly induced on the commutative monoid of processes by the
transition relation. The culprit is that, in general, two parallel
components may fuse into a single indecomposable process as a result
of scope extrusion. To define the decomposition order we consider a
fragment of the transition relation that avoids this phenomenon.
The second complication, which arises only in the case of weak bisimilarity, is
that certain transitions are deemed unobservable, and that, as a
consequence, there are transitions that do not change state (are
between weakly bisimilar processes). We demonstrate that a
decomposition order can, nevertheless, be obtained by ignoring such
\emph{stuttering transitions}.

The paper \cite{DELL16} studies unique parallel decomposition w.r.t. both
strong bisimilarity and weak bisimilarity for the applied $\pi$-calculus. 
The applied $\pi$-calculus is a variant of the $\pi$-calculus that was designed 
for the verification of cryptographic protocols.
Its main feature is that channels can only transmit variables and the values of the variables are
set using \emph{active substitutions}. Roughly, active substitution is an extension of the grammar of the $\pi$-calculus
that works as a `memory' that save the value of a variable.
% % % 
Because the variables in a transition are observable but the `memories' are not, it is possible to mask sensitive information.
The proof of the result for the strong case \newtext{in \cite{DELL16}} relies on induction over the \emph{norm} of a process
and the fact that the norm of the arguments of a parallel composition is less than the norm of the parallel composition.
Unfortunately, this property is not true because of the restriction operator
(see Section~\ref{sec:upd_strong_case} for a counter example).
This is the reason why we restrict ourselves to finite processes in the strong setting.
The proof of the weak case in \cite{DELL16} follows the proof technique attributed to Milner. 
The general techniques from \cite{LO05} cannot be applied directly in
the setting of the applied $\pi$-calculus due to the active
substitutions.

In \cite{Luttik16}, the second author presented an adaptation of the
general result of \cite{LO05} in order to make it suitable for
establishing unique parallel decomposition in settings with a notion
of unobservable behaviour.
The ensued technique amounts to showing that the transition relation
induces a so-called \emph{weak} decomposition order satisfying a
property that is called \emph{power cancellation}.
% The application of this technique is considerably more involved than
% applying the technique from \cite{LO05}.
%
In the present paper, we show how, instead of using the adapted
technique from \cite{Luttik16}, the original technique from
\cite{LO05} may be applied in settings with a notion of unobservable
behaviour, considering a stutter-free fragment of the transition
relation. This method appears to be simpler than the method
suggested by the result in \cite{Luttik16}.

The paper is organized as follows. In Section~\ref{sec:unique_decomposition}, we briefly recall the 
abstract framework introduced in~\cite{LO05} to prove UPD results.
In Section~\ref{sec:pi} we recall the syntax and different semantics of the $\pi$-calculus.
Section~\ref{sec:upd_strong_case} is composed of two subsections.
In Section~\ref{subsec:depth} we introduce the notion of \emph{depth} of a process and we prove some properties of this notion.
In Section~\ref{subsec:upd} we use these results and the result in Section~\ref{sec:unique_decomposition} 
to prove that finite processes of the $\pi$-calculus satisfy unique parallel decomposition w.r.t. strong bisimilarity. 
Section~\ref{sec:upd_weak_case} follows a similar structure.
In Section~\ref{subsec:pwss} we introduce the notion of \emph{processes without stuttering transitions} 
and we prove some properties of this kind of processes.
These properties and the result in Section~\ref{sec:unique_decomposition}
are used in Section~\ref{subsec:upd_for_wb} to prove that finite processes of the $\pi$-calculus satisfy unique parallel 
decomposition w.r.t.\ weak bisimilarity. 
In Section~\ref{sec:final_remarks} we present some final remarks.

\section{Decomposition orders}
\label{sec:unique_decomposition}

In this section, we briefly review the theory of unique decomposition for
commutative monoids that we shall apply in the remainder of the paper
to prove UPD results in the context of the $\pi$-calculus.
 
\begin{definition}
A \emph{commutative monoid} is a set $M$ with a distinguished element $e$ and a binary operation on $M$ 
denoted by $\cdot$ such that for all $x, y, z \in M$ :
\begin{itemize}
 \item 
 $x \cdot (y \cdot z) = (x \cdot y) \cdot z$ (associativity);
 \item 
 $x \cdot y = y \cdot x$ (commutativity);
 \item
 $x \cdot e = e \cdot x = x$ (identity).
\end{itemize}
\end{definition}
\noindent
In the remainder of the paper we often suppress the symbol $\cdot$ or use ${\parallel}$.

\begin{definition}
An element $p$ of a commutative monoid $M$ is called
\emph{indecomposable} if $p \neq e$ and $p = xy$ implies $x = e$ or $y
= e$. 
\end{definition}

\begin{definition}
Let $M$ be a commutative monoid. 
 A \emph{decomposition} in $M$ is a finite multi-set $\ms{p_1,\ldots , p_k}$ of indecomposable elements
 of $M$ such that $p_1 \cdot p_2 \cdots p_k$ is defined. 
 The element $p_1 \cdot p_2 \cdots p_k$ in $M$ will be called the \emph{composition associated} 
 with the decomposition $\ms{p_1 ,\ldots, p_k}$, and, conversely,
% % 
 we say that $\ms{p_1 ,\ldots, p_k}$ is a \emph{decomposition} of the element $p_1 \cdot p_2 \cdots p_k$ of $M$. 
% % 
 Decompositions $d = \ms{p_1 ,\ldots, p_k}$ and $d' = \ms{p'_1 ,\ldots, p'_l}$ are equivalent in $M$
 (notation $d \equiv d'$) if they have the same composition, i.e. $p_1 \cdot p_2 \cdots p_k  = p'_1 \cdots p'_l$.
 A decomposition $d$ in $M$ is \emph{unique} if $d \equiv d'$ implies $d = d'$ for all decompositions
 $d'$ in $M$. 
 We say that an element $x$ of $M$ has a \emph{unique decomposition} if it has a
 decomposition and this decomposition is unique. 
 If every element of $M$ has a unique decomposition, then
 we say that $M$ has \emph{unique decomposition}.
\end{definition}

Theorem~\ref{th:unique_decomposition_strong_case} below % and~\ref{th:unique_decomposition} 
gives a sufficient condition to ensure that a commutative monoid $M$
has unique decomposition. 
It requires the existence of a \emph{decomposition order} for $M$.
\begin{definition}
\label{def:wdo}
Let $M$ be a commutative monoid; a
partial order $\porder$ on $M$ is a \emph{decomposition order} if
\begin{enumerate}
 \item 
 it is \emph{well-founded}, i.e., for every non-empty subset $\hat M \subseteq M$
 there is $m \in \hat M$ such that for all $m'\in M$, $m' \porder m$ implies $m'=m$.  
 In this case, we say that $m$ is a $\porder$-minimal element of $\hat M$;
 \item 
 the identity element $e$ of $M$ is the least element of $M$ with respect to $\porder$, i.e.,
 $e \porder x$ for all $x$ in M;
 \item 
 \label{def:wdo:strict_compatible}
 ${\porder}$ is \emph{strictly compatible}, i.e., for all $x, y, z \in M$
 if $x \porderstrict y$ (i.e. $x \porder y$ and $x \neq y$) and $yz$ is defined, $xz \porderstrict yz$;
 \item 
 it is \emph{precompositional}, i.e., for all $x, y, z \in M$
 $x \porder yz$ implies $x = y'z'$ for some $y' \porder y$ and $z' \porder z$; and
 \item
 it is \emph{Archimedean}, i.e., for all $x, y \in M$
 $x^n \porder y$ for all $n \in \Nat$ implies that $x = e$.
\end{enumerate}
\end{definition}

\begin{theorem}[\cite{LO05}]
 \label{th:unique_decomposition_strong_case}
 Every commutative monoid $M$ with a decomposition order has unique decomposition.
\end{theorem}

\section{The $\pi$-calculus}
\label{sec:pi}

We recall the syntax of the $\pi$-calculus and  
the rules to define the transition relation~\cite{SW01}.
We assume a set of \emph{names} or \emph{channels} $\Var$. 
We use $a,b,c, x,y,z$ to range over $\Var$.
\begin{definition}
 %{\cite[Def. 1.1.1, P.11]{SW01}}
 The \emph{processes, summations} and \emph{prefixes} of the $\pi$-calculus
 are given respectively by 
 \begin{align*}
  P   ::= & \quad  M \quad | \quad P\parpi P' \quad | \quad \nu z. P \quad | \quad !P \\
  M   ::= & \quad \nullproc \quad | \quad \pi.P \quad | \quad  M + M' \\   
  \pi ::= & \quad \ov{x}y \quad | \quad x(z) \quad | \quad \tau \quad | \quad [x=y] \pi
 \end{align*}
 We denote with $\Pi$ the set of processes of the $\pi$-calculus. 
\end{definition}

\begin{table}
\begin{center}
 \begin{minipage}{0.2\textwidth}
 \begin{gather}
  \dedrule{}{\ov{x}y.P \trans[\ov{x}y] P}
  \tag{Out}
  \label{trans:out}
 \end{gather}
 \end{minipage}
 \begin{minipage}{0.25\textwidth}
 \begin{gather}
  \dedrule{}{x(z).P \trans[xy] P\{y/z\}}
  \tag{Inp}
  \label{trans:inp}
 \end{gather}
 \end{minipage}
 \begin{minipage}{0.18\textwidth}
 \begin{gather}
  \dedrule{}{\tau.P \trans[\tau] P}
  \tag{Tau}
  \label{trans:tau}
 \end{gather}
 \end{minipage}
 \begin{minipage}{0.25\textwidth}
 \begin{gather}
  \dedrule{\pi.P \trans[\alpha] P}{[x=x]\pi.P \trans[\alpha] P}
  \tag{Mat}
  \label{trans:mat}
 \end{gather}
 \end{minipage}
\\
 \begin{minipage}{0.24\textwidth}
 \begin{gather}
  \dedrule{P \trans[\alpha] P'}{P + Q \trans[\alpha] P'}
  \tag{Sum-L}
  \label{trans:sum-l}
 \end{gather}
 \end{minipage}
 \qquad
 \begin{minipage}{0.5\textwidth}
 \begin{gather}
  \dedrule{P \trans[\alpha] P'}{P \parpi Q \trans[\alpha] P' \parpi Q}%
  \quad \bn(\alpha) \cap \fn(Q) = \emptyset
  \tag{Par-L}
  \label{trans:par-l}
 \end{gather}
 \end{minipage}
\\
 \begin{minipage}{0.32\textwidth}
 \begin{gather}
  \dedrule{P \trans[\ov{x}y] P' \quad Q \trans[xy] Q'}{P \parpi Q \trans[\tau] P' \parpi Q'}%
  \tag{Comm-L}
  \label{trans:comm-l}
 \end{gather}
 \end{minipage}
\qquad
 \begin{minipage}{0.44\textwidth}
 \begin{gather}
  \dedrule{P \trans[\ov{x}(z)] P' \quad Q \trans[xz] Q'}{P \parpi Q \trans[\tau] \nu z(P' \parpi Q')} \ z \not \in \fn(Q)%
  \tag{Close-L}
  \label{trans:close-l}
 \end{gather}
 \end{minipage}
\\
 \begin{minipage}{0.37\textwidth}
 \begin{gather}
  \dedrule{P \trans[\alpha] P' }{ \nu z (P) \trans[\alpha] \nu z (P') } \ \ z \not \in n(\alpha)%
  \tag{Res}
  \label{trans:res}
 \end{gather}
 \end{minipage}
\begin{minipage}{0.31\textwidth}
 \begin{gather}
  \dedrule{P \trans[\ov{x}z] P' }{ \nu z (P) \trans[\ov{x}(z)] P' }
  \ z \neq x
  \tag{Open}
  \label{trans:open}
 \end{gather}
 \end{minipage}
 \begin{minipage}{0.28\textwidth}
 \begin{gather}
  \dedrule{P \trans[\alpha] P' }{!P \trans[\alpha] P' \parpi !P}%
  \tag{Rep-Act}
  \label{trans:rep-act}
 \end{gather}
 \end{minipage}
\\
 \begin{minipage}{0.4\textwidth}
 \begin{gather}
  \dedrule{P \trans[\ov{x}y] P' \quad P \trans[xy] P''}{!P \trans[\tau] (P' \parpi P'') \parpi !P}%
  \tag{Rep-Comm}
  \label{trans:rep-comm}
 \end{gather}
 \end{minipage}
 \begin{minipage}{0.55\textwidth}
 \begin{gather}
  \dedrule{P \trans[\ov{x}(z)] P' \quad P \trans[xz] P''}{!P \trans[\tau] (\nu z(P' \parpi P'')) \parpi !P} \ z \not \in \fn(P)%
  \tag{Rep-Close-L}
  \label{trans:rep-close-l}
 \end{gather}
 \end{minipage} 
 \end{center}
 \caption{Transition rules for the $\pi$-calculus}
 \label{tb:trans_rules}
\end{table}

An occurrence of a name $z \in \Var$ is \emph{bound} in a process $P$ 
if it is in the scope of a restriction $\nu z$ or of an input $a(z)$. 
A name $a \in \Var$ is \emph{free} in a process $P$ if there is at least one occurrence of $a$ that is not bound.
We write $\bn(P)$ and $\fn(P)$ to denote, respectively, the set of bound names and free names of a process~$P$. 
The \emph{set of names} of a process $P$ is defined by $n(P) = \bn(P) \cup \fn(P)$.
We employ the following convention of the $\pi$-calculus w.r.t. names.
\begin{convention}{\cite[P.47]{SW01}}
\label{conv:names2}
In any discussion, we assume that the bound names of any processes or actions under consideration are chosen
to be different from the names free in any other entities under consideration, such as processes, actions, substitutions,
and sets of names. This convention is subject to the limitation that in considering a transition $P \trans[\ov{x}(z)]Q$, 
the name $z$ that is bound in $\ov{x}(z)$ and in $P$ may occur free in $Q$.
This limitation is necessary for expressing scope extrusion.
\end{convention}

The transition relation associated to a term is defined by the rules in Table~\ref{tb:trans_rules},
where we have omitted the symmetric version of the rules 
(\ref{trans:sum-l}), (\ref{trans:par-l}), (\ref{trans:comm-l}) and (\ref{trans:close-l}).
We denote with $A$ the set of \emph{visible actions} that can be executed by a process $P\in \Pi$,
i.e. $A = \{xy \mid x,y \in \Var\} \cup \{\ov{x}y \mid x,y \in \Var\} \cup \{\ov{x}(z) \mid x,z \in \Var\}$.
The action $\tau$ is the \emph{internal action}. We define $A_\tau = A \cup \{\tau\}$. 
For $P, P' \in \Pi$ and $\alpha \in A_\tau$, we write $P \trans[\alpha] P'$ if there is a derivation of $P \trans[\alpha] P'$
with rules in Table~\ref{tb:trans_rules}.

\begin{definition}%[Strong bisimilarity]
\emph{Strong bisimilarity} is the largest symmetric relation over $\Pi$, notation~$\bisim$, 
such that whenever $P \bisim Q$, if $P \trans[\alpha] P'$ then there is $Q'$ s.t. $Q \trans[\alpha] Q'$ and $P' \bisim Q'$. 
\end{definition}

\noindent
The relation $\bisim$ is not compatible with input prefix: We have that
$$
\ov{z}x \parpi a(y) \bisim \ov{z}x.a(y) + a(y).\ov{z}x\enskip,
$$
whereas
$$
b(a).(\ov{z}x \parpi a(y)) \nbisim  b(a).(\ov{z}x.a(y) + a(y).\ov{z}),
$$
because when $z$ is received over the channel $a$, we have 
$$
(\ov{z}x \parpi a(y)) \{a/z\} \nbisim (\ov{z}x.a(y) +
a(y).\ov{z})\{a/z\}
\enskip.
$$
Hence, $\bisim$ is not a congruence for the full syntax of the
$\pi$-calculus. It is, however, a so-called \emph{non-input
  congruence} (see \cite[Theorem 2.2.8]{SW01}): it is compatible with
all the other constructs in the syntax. In the present paper
we shall only use the fact that $\bisim$ is compatible with parallel
composition, i.e., if $P_1 \bisim Q_1$ and $P_2 \bisim Q_2$, then
$P_1 \parpi P_2 \bisim Q_1 \parpi Q_2$.

\remarkBL{I propose that we remove the definitions of non-input
  congruence and of strong full bisimilarity and put a remark in Section~\ref{sec:final_remarks}. We
should say that we do not know how to extend the results to strong
full bisimilarity, except for in the case of the asynchronous
$\pi$-calculus where strong bisimilarity and strong full bisimilarity
coincide.}
\remarkMDL{Done!}
% \begin{definition}
%  A \emph{non-input} context is a context is which $[\cdot]$ does not occur under an input prefix. 
% %  
%  An equivalence relation $\relR$ is a \emph{non-input congruence} if $C[P] \relR C[Q]$ for all 
%  $P, Q$ s.t. $P \relR Q$ and all non-input contexts $C$.
% \end{definition}
% \medskip 
% \noindent
% 
% Notice $b(a).[\cdot]$ is not a non-input context.
% 
% This problem can be avoided strengthening the relation. 
% % 
% To do this, we introduce \emph{substitutions}.
% % 
% A \emph{substitution} $\sigma$ is a function of sort $\Var \to \Var$.
% % 
% Substitutions are extended to process in the standard way.
% % 
% In addition, see Convention~\ref{conv:names2} below. 
% % 
% We write $P\sigma$ to denote the application of $\sigma$ on $P$.
% 
% \begin{definition}
% $P$ and $Q$ are \emph{strong full bisimilar},
% notation $P \bisim^c Q$, if $P\sigma \bisim  Q\sigma$ for all substitutions $\sigma$.
% \end{definition}
% 
% Unfortunately, we were not able to prove the result of unique decomposition for the 
% strong full bisimulation. We point out the problems with this semantics in the Section~\ref{subsec:strong_full_bisimilarity}.
% %
 
 We recall now the weak variant of bisimilarity.
 We write $P \transs P'$ if $P = P'$ or if there are $P_0, \ldots, P_n$ with $n>0$
 s.t. $P = P_0 \trans[\tau] \ldots \trans[\tau] P_n = P'$.
 We write $P \transs[\alpha] Q$ with $\alpha \in A_\tau$ if there are $P', Q'$ s.t.
 $P \transs P' \trans[\alpha] Q' \transs Q$. 
 Notice the difference between $P \transs P'$ and $P \transs[\tau] P'$, in the second case, 
 at least one $\tau$-transition is executed.
%   in general $\alpha$ could be $\tau$ but 
%  this is not the case for Def.~\ref{def:wbisim}).

\begin{definition}
\label{def:wbisim}
\emph{Weak bisimilarity} is the largest symmetric relation over $\Pi$, notation~$\wbisim$, 
such that whenever $P \wbisim Q$, 
\begin{inparaenum}[(i)]
\item 
if $P \trans[\alpha] P'$ with $\alpha \in A$ then there is $Q'$ s.t. $Q \transs[\alpha] Q'$ and $P' \wbisim Q'$ and 
\item 
if $P \trans[\tau] P'$ then there is $Q'$ s.t. $Q \transs Q'$ and $P' \wbisim Q'$.
\end{inparaenum}
\end{definition}

Like strong bisimilarity, it is only possible to prove that $\wbisim$ is a congruence for
non-input contexts (see \cite[Theorem 2.4.22]{SW01}). 
% There also exists a full version of this semantics that is a congruence for every context.
% 

\section{Unique decomposition with respect to strong bisimilarity}
\label{sec:upd_strong_case}

In this section, we shall use the result presented in Section~\ref{sec:unique_decomposition} to prove 
that every \emph{finite} $\pi$-calculus process has a unique parallel decomposition w.r.t.\ strong bisimilarity. 
In Section~\ref{subsec:depth} we introduce the definition of depth of a process and some of its properties. 
We also explain why we restrict our development to finite
processes.
% \remarkBL{Do we?}
% 
In Section~\ref{subsec:upd} we present the unique decomposition result.

\subsection{The depth of a process}
\label{subsec:depth}

Given a set $X$, we denote with $X^*$ the set of finite sequences over $X$,
where $\varepsilon \in X^*$ is the empty sequence.
For $\omega = \alpha_1\alpha_2\cdots \alpha_n \in A_\tau^*$ with $n>0$, 
we write $P \trans[\omega] P'$ if there are processes
$P_0, P_1, \ldots, P_n$ s.t.  $P = P_0 \trans[\alpha_1] P_1 \trans[\alpha_2] \ldots \trans[\alpha_n] P_n = P'$. 
If $\omega = \varepsilon$, then $P \trans[\omega] P'$ implies $P'=P$.
If we are not interested in $P'$, we write $P \trans[\omega]$.
In addition, we write $P\stops$ if for all $\alpha \in A_\tau$, $P \ntrans[\alpha]$.

\begin{definition}
 Let $\length : A_\tau^* \to \Nat$ be the function defined by 
 \begin{equation*}
 \length(\omega) =
\begin{cases}
0 & \text{if } \omega = \varepsilon,\\
\length(\omega') + 1 & \text{if } \omega = \alpha\omega' \text{ and } \alpha \neq \tau\\
\length(\omega') + 2 & \text{if } \omega = \alpha\omega' \text{ and } \alpha = \tau
\end{cases}
\end{equation*}
\end{definition}

\begin{definition}
 A process $P \in \Pi$ is \emph{normed} if there is $\omega \in A_\tau^*$ such that 
 $P \trans[\omega] P' \stops$.  
 We denote with $\Pi_n$ the set of normed processes.
 The $\textit{depth} : \Pi_n \to \Nat \cup \{\infty\}$ and the 
 $\textit{norm} : \Pi_n \to \Nat$ of a normed process $P \in \Pi$ are defined, respectively, by 
 \begin{align*}
 \depth P  = & 
 \sup(\{\length(\omega) \mid P \trans[\omega] P' \text{ and } P' \stops\})
 \\
 \norm P = &
 \inf(\{\length(\omega) \mid P \trans[\omega] P' \text{ and }  P' \stops\})
\end{align*}
Where $\sup(X) = \infty$ whenever $X$ is an infinite set, and $\inf(\emptyset)=\infty$.
\end{definition}

We remark that we have assigned a higher weight to
occurrences of the label $\tau$ in the definition of the length of a
sequence $\omega\in A_\tau^*$. This is to ensure that depth is
additive w.r.t.\ parallel composition (i.e., the depth of a parallel
composition is the sum of the depths of its components), as we shall
prove in Lemma~\ref{lemma:sum_of_depths}) below.
% 
%In \cite{LO05}, this is done explicitly by means of \emph{linear communication functions}. 
% The unique decomposition results in \cite{LO05,DELL16} rely
% on norm being additive 
% For example, the norm of a normed parallel composition is equal to the sum of the norms of its arguments.
% (To ensure this kind of properties the weight of the label $\tau$ is 2:
% a $\tau$-transition could be the result of synchronizing two visible actions of weight 1).
% 
As opposed to other process calculi for which unique decomposition
has been established (see, e.g. \cite{LO05}), due to scope extrusion,
norm is not additive for the $\pi$-calculus:
 Consider 
 $$P = P_0 \parpi P_1 = \nu z(\ov{a}z) \parpi a(x).!\ov{x}a$$
 $P$~is normed because $P \trans[\tau] \nu z(\nullproc \parpi !\ov{z}a) \stops$
 but $P_1$ is not because it only performs an execution of infinite length. 
% % 
 Then, to ensure this kind of properties, one approach could be just consider normed processes. 
% % 
 Unfortunately this is not enough. 
 Consider 
$$Q = Q_0 \parpi Q_1 = \nu z(\ov{a}z) \parpi a(x).\ov{x}a$$ 
Processes $Q$, $Q_0$ and $Q_1$ are normed and, moreover, they perform no infinite execution. Despite this, 
we have that $\norm Q = 2$, because $Q \trans[\tau] \nu z(\nullproc \parpi \ov{z}a) \stops$,
and $\norm{Q_0} + \norm{Q_1} = 1 + 2 = 3$.
Moreover, notice that the norm of the arguments of a parallel composition is not less than the norm of
the parallel composition, i.e. $\norm Q = 2$ and $\norm{Q_1} = 2$.
\newtext{In particular, these examples show that item 4 in Lemma 3 of \cite{DELL16} 
(norm is additive) is false, and, as a consequence, some proofs in \cite{DELL16} are flawed. 
The authors of \cite{DELL16} proposed a solution to this problem that we discuss in the conclusion of this paper.}
So, to facilitate inductive reasoning, we will consider \emph{finite processes} and depth. 
 \begin{definition}
  A process $P \in \Pi$ is \emph{finite} if there is $n \in \Nat$ s.t. there is no 
  $\omega = \alpha_1\alpha_2\cdots\alpha_{n+1} \in A_\tau^*$ such that $P \trans[\omega]$.
  We denote with $\finitePi$ the set of finite processes of $\Pi$.
  \end{definition}
Following the last example, we have that 
$Q$, $Q_0, Q_1 \in \finitePi$ and $\depth{Q} = 3 = \depth{Q_0} + \depth{Q_1}$.

\medskip
To conclude this section we present a collection of results including lemmas and theorems. 
Most of the lemmas are only needed to prove the theorems. 
The theorems and only few lemmas will be used in the next section.
Theorem~\ref{th:bisim_implies_same_depth} states that bisimilar processes have the same depth. 
Theorem~\ref{th:parameters_depths_lesser_than_par_depth_strong_case} states that 
the depth of a parallel composition of two processes not bisimilar to $\nullproc$ 
is greater than the depth of each process.
Thanks to these results, we will able to extend the notion of depth to equivalence classes
and apply inductive reasoning. 
% (Proofs of Lemmas~\ref{lemma:PnbisimNullproc} and~\ref{lemma:trans_reduces_depth}
% are in the Appendix.)

\begin{lemma}
 \label{lemma:PnbisimNullproc}
 For all $P \in \finitePi$, $P \nbisim \nullproc$ implies $\depth P > 0$. 
\end{lemma}
% \begin{proof}
%  If $P \nbisim \nullproc$ then there is $\alpha \in A_\tau$ s.t. $P \trans[\alpha] P'$,
% %  . In addition, $P \in \finitePi$ implies $P$ executes no  infinite execution.
%  and therefore $\depth P > 0$.
% \end{proof}

\begin{lemma}
\label{lemma:trans_reduces_depth}
If $P \in \finitePi$ and $P\trans[\alpha] P'$, $\alpha \in A_\tau$, then $P'\in \finitePi$ and $\depth P > \depth{P'}$. 
\end{lemma}
% \begin{proof}
%  If $P' \not \in \finitePi$, then for every $n\in\Nat$ there exists
%  $\omega' = \alpha_1\alpha_2\cdots\alpha_{n} \in A_\tau^*$ such that
%  $P' \trans[\omega']$. Since $P\trans[\alpha] P'$, it then follows that
%  for every $n\geq 1$ there exists
%  $\omega=\alpha\alpha_1\alpha_2\cdots\alpha_{n-1}\in A_\tau^*$ such
%  that $P\trans[\omega]$, and hence $P \not\in \finitePi$. 
%  So, by contraposition, if $P\in\finitePi$, then also $P'\in\finitePi$.
%  In addition, $\depth{P} \geq \length(\alpha) + \depth{P'}$ and $\length(\alpha) > 0$ imply
%  $\depth{P} > \depth{P'}$.
% \end{proof}

\begin{theorem}
% \label{lemma:bisim_implies_same_depth}
\label{th:bisim_implies_same_depth}
If $P \bisim Q$ then $P \in \finitePi$ iff $Q \in \finitePi$;
moreover, $\depth{P} = \depth{Q}$. 
\end{theorem}
\begin{proof}
 Suppose that $P\bisim Q$. Then, clearly, $P\trans[\omega]$ iff
 $Q\trans[\omega]$, and hence $P\in\finitePi$ iff $Q\in\finitePi$.

 To prove that $\depth{P} = \depth{Q}$, first note that if
 $P,Q\not\in\finitePi$, then $\depth{P}=\infty=\depth{Q}$.
 In the case that remains, both $\depth{P}$ and $\depth{Q}$ are
 natural numbers; we proceed by induction over $\depth{P}$.
 If $\depth{P} = 0$ then $P \bisim \nullproc$ by Lemma~\ref{lemma:PnbisimNullproc} , so $Q \bisim \nullproc$ and therefore $\depth{Q} = 0$.
 Suppose now $\depth{P} = n > 0$.
 Assume that the statement holds for processes with depth less than $n$. 
 Suppose $\depth{Q} = m > n$ then there is $Q'$ s.t. $Q \trans[\alpha] Q'$ and $m = \length(\alpha) + \depth{Q'}$. 
 By definition of ${\bisim}$, we get $P \trans[\alpha] P'$, $P' \bisim Q'$. 
 By Lemma~\ref{lemma:trans_reduces_depth}, $\depth{P'} < \depth{P}$, therefore $\depth{P'} = \depth{Q'}$ 
 and $\depth{P} \geq \depth{P'} + \length(\alpha) = m > n = \depth{P}$, i.e. we get a contradiction.
 Similarly, for the case $\depth{Q} = m < n$, 
 we can reach a contradiction by considering a transition  
 $P \trans[\alpha] P'$ with $n = \length(\alpha) + \depth{P'}$. 
 Then we can conclude $\depth{P} = \depth{Q}$.
%  
%  Given that $\length(\alpha) = \length(\alpha)$, we get a contradiction with $\depth{P} < \depth{Q}$.
\end{proof}

\begin{lemma}
\label{lemma:trans_and_max}
Let $P, P' \in \finitePi$, $\omega = \alpha\omega' \in A_\tau^*$ be such that 
$P \trans[\alpha] P' \trans[\omega']$ and $\depth{P} = \length(\omega)$. 
Then $\depth{P'} = \length(\omega')$ and therefore $\depth{P} = \depth{P'} + \length(\alpha)$.
\end{lemma}
% \begin{proof}
% Clearly  $\depth{P'} \geq \length(\omega')$.
% % 
% If $\depth{P'} > \length(\omega')$ then $\depth{P} < \depth{P'} + \length(\alpha)$ and we get a contradiction 
% with $\depth{P} = \length(\omega)$. 
% \end{proof}

\begin{lemma}
 \label{lemma:max_and_rest}
 For all $P \in \finitePi$, $\depth{P} \geq \depth{\nu z(P)}$ for all $z \in \Var$.
\end{lemma}
% \begin{proof}
%   The proof follows from the fact that whenever $\nu z(P) \trans[\omega]$ for $\omega \in A_\tau^*$, 
%   then there is $\omega' \in A_\tau^*$ s.t. $P \trans[\omega']$ and $\length(\omega') \geq \length(\omega)$.
% %   
%   The proof of this auxiliary fact follows by induction on $\length(\omega)$ and case analysis 
%   taking into account rules (\ref{trans:res}) and (\ref{trans:open}).
% \end{proof}

\begin{lemma}
 \label{lemma:sum_of_depths_aux}
 Let $P, Q \in \finitePi$ and $\omega \in A_\tau^*$
 be such that $P  \parpi Q \trans[\omega]$
 and $\length(\omega) = \depth{P\parpi Q}$. 
 Then, there are $\omega_1, \omega_2 \in A_\tau^*$ such that $P \trans[\omega_1] $, $Q \trans[\omega_2] $ and 
 $\length(\omega_1) + \length(\omega_2) = \length(\omega)$.
\end{lemma}
\begin{proof}
%We proceed by complete induction on $n = \depth{P\parpi Q}$.
We proceed by complete induction on $n = \depth{P} + \depth{Q}$.
Suppose that the property holds for parallel compositions of finite
processes such that the sum of the depths is smaller than $n > 0$.
Let $\omega = \alpha\omega'$ and $R$ be such that 
$\length(\omega)=n$ and 
$P\parpi Q \trans[\alpha] R \trans[\omega']$. 
We analyse the different ways of deriving the first transition (we omit the symmetric cases).
\begin{itemize}
 \item Case (\ref{trans:par-l}). Then $P \trans[\alpha] P'$ and $R = P' \parpi Q$.
 By Lemma~\ref{lemma:trans_and_max} and induction $\depth{P' \parpi Q} = \length(\omega')< n$
 and then there are $\omega_1$ and $\omega_2$ s.t. $P' \trans[\omega_1]$,
 $Q \trans[\omega_2]$ and $\length(\omega_1) + \length(\omega_2) = \length(\omega')$.
 Then $P \trans[\alpha\omega_1]$ and 
 $\length(\alpha\omega_1) + \length(\omega_2) = \length(\omega)$.
 
 \item 
 Case (\ref{trans:comm-l}). Then $P \trans[\ov{x}y] P'$, $Q \trans[xy] Q'$ and $R = P' \parpi Q'$ 
 and $\alpha = \tau$.
 By Lemma~\ref{lemma:trans_and_max} and induction 
 $\depth{P' \parpi Q'} = \length(\omega') < \length(\tau) + \length(\omega') = n$
 and then there are $\omega_1$ and $\omega_2$ s.t. $P' \trans[\omega_1]$,
 $Q' \trans[\omega_2]$ and $\length(\omega_1) + \length(\omega_2) = \length(\omega')$.
 Then $P \trans[\ov{x}y\omega_1]$ and $Q \trans[xy\omega_2]$ and 
 $\length(\ov{x}y\omega_1) + \length(xy\omega_2) = 
 \length(\tau) + \length(\omega') = \length(\omega)$.
  
 \item
 Case (\ref{trans:close-l}). Then $P  \trans[\ov{x}(z)] P'$,$Q   \trans[xz]Q'$, $\alpha=\tau$ and $R = \nu z(P'\parpi Q')$.
 The side condition of (\ref{trans:close-l}) allows us to use the rules (\ref{trans:par-l}) and its symmetric version, 
 then $P \parpi Q \trans[\ov{x}(z)xz] P' \parpi Q'$.
  On one hand, by Lemma~\ref{lemma:max_and_rest}, $\depth{P' \parpi Q'} \geq \depth{\nu z(P'\parpi Q')}$.
  On the other hand, $\depth{P' \parpi Q'} \leq \depth{\nu z(P'\parpi Q')}$
  because $\omega = \tau\omega'$ is a maximal execution, 
  $\length(\tau) = \length(\ov{x}zxz)$, 
  $\nu z(P'\parpi Q') \trans[\omega']$ and by Lemma~\ref{lemma:trans_and_max}.
  Then $\depth{P' \parpi Q'} = \depth{\nu z(P'\parpi Q')} < n$.
  Moreover, there is $\omega''$ such that $P' \parpi Q' \trans[\omega'']$ and $\length(\ov{x}zxz\omega'') = \depth{P \parpi Q}$.
  Then $P  \parpi Q  \trans[\ov{x}(z)] P' \parpi Q  \trans[xz\omega'']$ with 
  $\length(xz\omega'') < n$. From this point we can repeat the proof of the first case.
\end{itemize}
\end{proof}

\begin{lemma}
 \label{lemma:sum_of_depths}
 For all processes $P, Q \in \finitePi$, $\depth{P\parpi Q} = \depth P + \depth Q$.
\end{lemma}
\begin{proof}
By Lemma~\ref{lemma:sum_of_depths_aux} we can ensure $\depth{P\parpi Q} \leq \depth P + \depth Q$.
On the other hand, by Convention~\ref{conv:names2}, we have that  
$P  \trans[\omega]$ or $Q \trans[\omega]$ implies $P  \parpi Q \trans[\omega]$.
This allows us to conclude $\depth{P\parpi Q} \geq \depth P + \depth Q$ and therefore 
$\depth{P\parpi Q} = \depth P + \depth Q$.
\end{proof}

 \begin{lemma}
  \label{lemma:finiteness_is_conservative}
  For all $P, Q \in \Pi$,  $P, Q \in \finitePi$ iff $P \parpi Q \in \finitePi$.
 \end{lemma}
%  \begin{proof}
%  The left to right implication is a a straight consequence of Lemma~\ref{lemma:sum_of_depths}. 
%  For the other implication, w.l.o.g. assume that $\length(P) = \infty$, 
%  by rule (\ref{trans:par-l}) $\length(P \parpi Q) = \infty$ and we get a contradiction.
%  \end{proof}

\begin{theorem}
%  \label{lemma:parameters_depths_lesser_than_par_depth_strong_case}
 \label{th:parameters_depths_lesser_than_par_depth_strong_case} 
 If $P, Q, R \in \finitePi$, $P \nbisim \nullproc$, $Q \nbisim \nullproc$ and 
 $P \parpi Q \bisim R$ then $\depth P < \depth R$ and $\depth Q < \depth R$.
\end{theorem}
\begin{proof}
 By Lemma~\ref{lemma:PnbisimNullproc}, $\depth P > 0$, $\depth Q > 0$.
 By Theorem~\ref{th:bisim_implies_same_depth}, $\depth R = \depth{P \parpi Q}$.
 By Lemma~\ref{lemma:sum_of_depths},
 $\depth R = \depth{P} + \depth Q$ and
 we conclude $\depth P < \depth R$ and $\depth Q < \depth R$.
\end{proof}

\subsection{Unique decomposition}
\label{subsec:upd}

The commutative monoid associated with $\finitePi$ modulo $\bisim$ is defined by 
\begin{itemize}
 \item $\spp = \{[P]_{\bisim}: P \in \finitePi\}$ where $[P]_{\bisim} = \{P' : P' \bisim P\}$
 \item $e = \sec{\nullproc} \in \spp$ .
 \item ${\parallel} : \spp \times \spp \to \spp$ is such that $\sec{P} \parallel \sec{Q} = \sec{P \parpi Q}$
\end{itemize}
% The overloading of ${\parallel}$ will be harmless. 
By Lemma~\ref{lemma:finiteness_is_conservative} we have that the definition of $\parallel$ is sound.
By Theorem~\ref{th:bisim_implies_same_depth} we have that all $P' \in \sec{P}$
have the same depth. Then we can lift function \textit{depth} to $\spp$, i.e. $\depth{\sec{P}} = \depth{P}$.

\begin{lemma}
$\spp$ with neutral element $\sec{\nullproc}$ and binary operation {$\parallel$}
is a commutative monoid.
I.e., ${\parallel} \subseteq \spp \times \spp$ satisfies the associativity, commutativity and identity properties.
\end{lemma}

In order to use the Theorem~\ref{th:unique_decomposition_strong_case}
we need to define on $\spp$ a decomposition order.
In \cite{LO05,Luttik16}, it is shown that the transition relation
directly induces a decomposition order on a commutative monoid of processes.
In the case of the $\pi$-calculus, however, the order induced by the
transition relation cannot be directly used, as is illustrated by the
following example: Define a binary relation ${\rightsquigarrow} \subseteq \spp \times \spp$ by 
$\sec{R} \rightsquigarrow \sec{S}$ if there is $R' \in \sec{R}$ and $S' \in \sec{S}$ such that 
$R' \trans[\alpha] S'$.	We denote the inverse of the reflexive-transitive closure of ${\rightsquigarrow}$ 
by ${\porder_{\rightsquigarrow}}$, i.e., ${\porder_{\rightsquigarrow}} = ({\rightsquigarrow}^*)^{-1}$.
The order ${\porder_{\rightsquigarrow}}$ is not precompositional. 
Consider the processes $P = \nu z.(\ov{a}z.\ov{z}c.\ov{c}a)$ and $Q = a(x).x(y).\ov{y}b$. 
Then 
\begin{itemize}
 \item 
  $P \parpi Q = \nu z.(\ov{a}z.\ov{z}c.\ov{c}a) \parpi a(x).x(y).\ov{y}b 
  \trans[\tau] \nu z.(\ov{z}c.\ov{c}a \parpi z(y).\ov{y}b) = R$ and therefore 
 \item $\sec{P} \parallel \sec{Q} = \sec{P \parallel Q} \rightsquigarrow 
        \sec{R} = \sec{\nu z.(\ov{z}c.\ov{c}a \parpi z(y).\ov{y}b)}$.
 \item \remarkBL{Note that we should actually argue that $\nu
     z(\ov{c}a \parpi \ov{c}b)$ is indecomposable.}
%      \remarkMDL{Why? For me it is clear that $R$ is indecomposable because it executes only one transition.
%      After this transition, }
  Note that $R$ executes only one transition, i.e.  
  $\nu z.(\ov{z}c.\ov{c}a \parpi z(y).\ov{y}b) \trans[\tau] \nu z(\ov{c}a \parpi \ov{c}b)$,  then
  it is clear that there are no processes $P'$ and $Q'$ s.t. 
  $$\sec{\nu z.(\ov{a}z.\ov{z}c.\ov{c}a)} \rightsquigarrow^* \sec{P'} \qquad 	
  \sec{a(x).x(y).\ov{y}b} \rightsquigarrow^* \sec{Q'} \qquad 
  \sec{P'} \parallel \sec{Q'} = \sec{\nu z(\ov{z}c.\ov{c}a \parpi z(y).\ov{y}b)}$$ 
\end{itemize}

The particularity of this example is the \emph{scope extrusion}.
We need to define an order based on a fragment of the transition relation that avoids this phenomenon.
We shall define the partial order ${\porder}$ over $\spp$ as the
reflexive-transitive closure of the relation ${\trans} \subseteq \spp
\times \spp$, which is, in turn, defined as follows: %
\begin{align*}
 {\trans_0} & = 
 \{(\sec{P},\sec{Q}):  P \trans[\alpha] Q, \alpha \in A_\tau \text{ and } \not \exists 
 P_0, P_1 \in \finitePi \text{ s.t. } P_0 \nbisim \nullproc,  P_1 \nbisim \nullproc, P_0 \parpi P_1 \bisim P \} \\[1ex]
{\trans_{k+1}}  
 & = \{(\sec{P_0 \parpi P_1}, \sec{Q_0 \parpi P_1} ): \sec{P_0} \trans_k \sec{Q_0}, P_1 \in \finitePi\} \\
 &\cup \{(\sec{P_0 \parpi P_1},\sec{P_0 \parpi Q_1}) : \sec{P_1} \trans_k \sec{Q_1}, P_0 \in \finitePi\} \\[1ex]
{\trans} & = \bigcup_{k=0}^{\infty} \trans_{n}\enskip.
\end{align*}

The partial order ${\porder}$ is defined as the inverse of the reflexive-transitive closure of ${\trans}$  
i.e., ${\porder} = ({\trans}^*)^{-1}$. We write $\sec{P} \prec \sec{Q}$ if
$\sec{P} \porder \sec{Q}$ and $\sec{P} \neq \sec{Q}$.
Notice that the definition of $\trans$  avoids any kind of communications between the arguments of the parallel operator,
this ensures that the scope extrusion is also avoided.

\begin{lemma}
 \label{lemma:trans_reduces_depth_sec}
 If $\sec{P} \trans \sec{Q}$ then $\depth{\sec Q} < \depth{\sec P}$.
\end{lemma}
% \begin{proof}
%  By Lemma~\ref{lemma:trans_are_real_trans}, 
%  there are $\tilde P \in \sec{P}$ and $\tilde Q \in \sec{Q}$ s.t. $\tilde P \trans[\alpha] \tilde Q$ 
%  with $\alpha \in A_\tau$. 
%  By definition of $\spp$, $\tilde P \in \finitePi$ and 
%  by Lemma~\ref{lemma:trans_reduces_depth}, $\depth {\tilde Q} < \depth{\tilde P}$ and therefore 
%  $\depth{\sec Q} < \depth{\sec P}$.
% \end{proof}

\begin{lemma}
\label{lemma:partial_order_strong_case}
 ${\porder}$ is a partial order. 
\end{lemma}

\begin{proof}
 We have to prove that ${\porder}$ is reflexive, antisymmetric, and transitive. 
 ${\porder}$ is reflexive and transitive because it is the reflexive-transitive closure of ${\trans}$.
 To prove that ${\porder}$ is antisymmetric notice that $\sec{P} \prec \sec{Q}$ implies
 $\sec{Q} = \sec{P_n} \trans \ldots \trans \sec{P_1} \trans \sec{P_0} = \sec{P}$ for $n > 0$
 and then, by Lemma~\ref{lemma:trans_reduces_depth_sec}, $\depth{\sec{P}} < \depth{\sec{Q}}$.
 Therefore $\sec{P} \porder \sec{Q}$ and $\sec{Q} \porder \sec{P}$ implies $\sec{P} = \sec{Q}$.
\end{proof}

In Lemma~\ref{lemma:dec_order_strong_case}, we prove that $\porder$ is a decomposition order. 
To prove this result we need to add a last auxiliary result, Lemma~\ref{lemma:there_is_always_a_path}. 
% Before proving that ${\porder}$ is a decomposition order we introduce some technical lemmas
% that state properties of $\trans$.

\begin{lemma}
 \label{lemma:there_is_always_a_path}
 If $P \in \finitePi$ and $\depth P > 0$ then there is $Q$ s.t. $\sec{P} \trans \sec{Q}$.
\end{lemma}
\begin{proof}
We proceed by complete induction over $n = \depth P$. 
Assume that the hypothesis holds for values less than $n \geq 1$.    
Suppose there are no $P_0, P_1 \in \finitePi$ such that $P \bisim P_0 \parpi P_1$, $P_0 \nbisim \nullproc$
and $P_1 \nbisim \nullproc$. 
Given that $n \geq 1$ then there is $\alpha \in A_\tau$ s.t. $P \trans[\alpha] P'$.
Then $\sec P \trans_0 \sec{P'}$ and therefore $\sec P \trans \sec{P'}$. Finally, we can define $\sec{Q} = \sec{P'}$.

Suppose there are $P_0, P_1 \in \finitePi$ such that $P \bisim P_0 \parpi P_1$, $P_0 \nbisim \nullproc$
and $P_1 \nbisim \nullproc$, then $\sec{P} = \sec{P_0 \parpi P_1}$. 
By Theorem~\ref{th:parameters_depths_lesser_than_par_depth_strong_case}
$\depth{P_0} < \depth{P}$ and by Lemma~\ref{lemma:PnbisimNullproc} $\depth{P_0} > 0$.
By induction there is $Q_0$ s.t. $\sec{P_0} \trans \sec{Q_0}$ and therefore there is $k$ s.t.  
$\sec{P_0} \trans_k \sec{Q_0}$
By definitions of ${\trans_{k+1}}$ and $\trans$, 
$\sec{P_0 \parpi P_1} \trans_{k+1} \sec{Q_0 \parpi P_1}$ and 
$\sec{P_0 \parpi P_1} \trans \sec{Q_0 \parpi P_1}$. 
Therefore if we define $\sec{Q} = \sec{Q_0 \parallel P_1}$ the proof is complete.
\end{proof}

\begin{lemma}
 \label{lemma:dec_order_strong_case}
 ${\porder} \subseteq \spp \times \spp$ is a decomposition order.
\end{lemma}

\begin{proof}
 \begin{enumerate}
  \item ${\porder}$ is well-founded. %Given arbitrary non-empty set there is always a $\porder$-minimal element.
   We have to prove that every non-empty subset of $\spp$ has a $\porder$-minimal element. 
  Let $X \subseteq \spp$ with $X \neq \emptyset$. 
  Let $\sec{P}$ be s.t. $\sec{P} \in X$ and $\depth{\sec{P}} = \min\{\depth{\sec{Q}} \mid \sec{Q} \in X\}$,
  then $\sec{P}$ is a minimal element of $X$
  by Lemma~\ref{lemma:trans_reduces_depth_sec} and definition of $\porder$.

  \item $\sec{\nullproc}$ is the least element of $\spp$ w.r.t. ${\porder}$.
  We consider $\sec{P}$ and proceed by induction on $\depth P$.
  If $\depth P = 0$, then $P \bisim \nullproc$ and therefore $\sec{P} = \sec{\nullproc}$.
  Suppose that $\depth P = n > 0$. 
  By Lemma~\ref{lemma:there_is_always_a_path} there is $Q$ s.t. $\sec{P} \trans \sec{Q}$.
  By Lemma~\ref{lemma:trans_reduces_depth_sec}, $\depth{Q} < \depth{P}$.
  By induction and definition of~$\porder$, $\sec{\nullproc} \porder \sec{Q}  \porder \sec{P}$.

  \item ${\porder}$ is strictly compatible. 
  Suppose $\sec{Q} \prec \sec{P}$ and consider $\sec{P} \parallel \sec{S}$. 
  By definition of ${\prec}$ there are $P_0, \ldots, P_n \in \finitePi$, with $n>0$,
  s.t. $\sec{P} = \sec{P_0} \trans \sec{P_1} \trans \ldots \trans \sec{P_n} = \sec{Q}$.
  By definition of $\trans$, for each $i = 0, \ldots n-1$ there is $k_i$ s.t. 
  $\sec{P_i} \trans_{k_i} \sec{Q_i}$. Define $k = \max\{k_i : i = 0, \ldots n-1\}$. Then
  $$
   \sec{P} \parallel \sec{S} = \sec{P_0} \parallel \sec{S} = \sec{P_0 \parpi S}
   \trans_{k+1} \ldots \trans_{k+1}
   \sec{P_n \parpi S} = \sec{P_n} \parallel \sec{S} = \sec{Q} \parallel \sec{S}
  $$
  By definition of $\trans$,  
  $$
   \sec{P} \parallel \sec{S} = \sec{P_0} \parallel \sec{S} = \sec{P_0 \parpi S}
   \trans \ldots \trans
   \sec{P_n \parpi S} = \sec{P_n} \parallel \sec{S} = \sec{Q} \parallel \sec{S}
  $$
  By Lemma~\ref{lemma:trans_reduces_depth_sec} and $n>0$, 
  $\depth{\sec{P} \parallel \sec{S}} > \depth{\sec{Q} \parallel \sec{S}}$. 
  Then ${\sec{Q} \parallel \sec{S}} \prec {\sec{P} \parallel \sec{S}} $.
  \item ${\porder}$ is precompositional. 
  Suppose $\sec{P} \porder \sec{Q} \parallel \sec{R}$, we have to prove there are
  $\sec{Q'} \porder \sec{Q}$ and $\sec{R'} \porder \sec{R}$ s.t. 
  $\sec{P} = \sec{Q'} \parallel \sec{R'}$.
  If $Q \bisim R \bisim \nullproc$ then $Q' \bisim R' \bisim \nullproc$ and the conditions are satisfied.
  Suppose that only one of both processes is bisimilar to $\nullproc$. 
  W.l.o.g. suppose $Q \nbisim \nullproc$ and $R \bisim \nullproc$. 
  In this case, $\sec{P} \porder \sec{Q} \parallel \sec{R} = \sec{Q}$, then 
  if we define  $\sec{Q'} = \sec{P}$ and $\sec{R'} = \sec{\nullproc}$, the conditions are also satisfied.
  Suppose now that $Q \nbisim \nullproc$ and $R\nbisim \nullproc$.  
  By definition of ${\porder}$ there are $n \geq 0$ and processes $S_n, \ldots, S_0$ s.t.
  $\sec{Q} \parallel \sec{R} = \sec{Q \parpi R} = \sec{S_n} \trans \ldots 
  \trans \sec{S_0} = \sec{P}$. 
  The proof proceed by induction on $n$. Suppose that the hypothesis holds for $n$, we prove the case $n+1$.
  Given that 
  $\sec{S_{n+1}} = \sec{Q \parpi R} = \sec{Q} \parallel \sec{R} \trans \sec{S_{n}}$,
  by definition of $\trans$, there is $T$ s.t. either
  $\sec{Q} \trans \sec{T}$ and $\sec{S_n} = \sec{T \parpi R}$ or,
  $\sec{R} \trans \sec{T}$ and $\sec{S_n} = \sec{Q \parpi T}$.
  (We have omitted the sub-index of {$\trans$} because it does not play any role.)
  W.l.o.g. suppose that $\sec{Q} \trans \sec{T}$ and $\sec{S_n} = \sec{T \parpi R}$.
  Then $\sec{P} \porder \sec{T \parpi R} = \sec{T} \parallel \sec{R}$.
  By induction there are $\sec{T'}$ and $\sec{R'}$ s.t. 
  $\sec{T'} \porder \sec{T}$, $\sec{R'} \porder \sec{R}$ and
  $\sec{P} = \sec{T'} \parallel \sec{R'}$.
  Because $\sec{T} \porder \sec{Q}$ and $\porder$ is a partial order, we have that 
  $\sec{T'} \porder \sec{Q}$ and we can conclude the proof.

  \item ${\porder}$ is Archimedean. 
  Suppose that $\sec{P}, \sec{Q} \in \spp$ are s.t.
  $\sec{P}^n \porder \sec{Q}$ for all $n \in \Nat$. 
  By Lemma~\ref{lemma:sum_of_depths}, $\depth{P^n} = n \cdot \depth P$. 
  Given that $\depth Q \in \Nat$ we can conclude that $\depth P = 0$ and therefore $\sec{P} = \sec{\nullproc}$.
 \end{enumerate}
\end{proof}

By Theorem~\ref{th:unique_decomposition_strong_case}, it follows that $\spp$ has unique decomposition.

\begin{corollary}
 The commutative monoid $\spp$ has unique decomposition.
\end{corollary}

\section{Unique parallel decomposition with respect to weak bisimilarity}
\label{sec:upd_weak_case}

To prove the result of unique parallel decomposition w.r.t. strong bisimilarity,
we relied on the definition of depth and on the properties that are satisfied when we take into account 
strong bisimilarity. 
In particular, we proved that all strongly bisimilar processes have the same depth.
For the weak bisimilarity we do not have the same property.
Consider the following processes 
$$
P = \ov{x}y.\nullproc \qquad P' = \tau.\ov{x}y.\nullproc \qquad P'' = \tau.\tau.\ov{x}y.\nullproc
$$
Notice that $P \wbisim P' \wbisim P''$, despite this, $\depth P < \depth {P'} < \depth {P''}$.
% 
% Then we cannot repeat straightforwardly the previous developments.
% 
To avoid this problem and to adapt the ideas behind results in the previous section,
we will consider processes without \emph{stuttering transitions}.
A transition $P \trans[\alpha] P'$ is a stuttering transition if $\alpha = \tau$ and $P \wbisim P'$. 

We could not establish UPD for normed processes in the strong setting,
because the norm of the arguments of a parallel composition is not necessarily less 
than the norm of the parallel composition.
%  parallel components of a normed parallel composition need not
% be normed themselves. 
% 
In the weak setting, it is known that normed processes of $\Pi$ do not satisfy UPD w.r.t.\
bisimilarity. Consider the following counter example \cite{DELL16}:
define $P = \nu z(\ov{z}c \parpi z(x).!\ov{a}b \parpi z(y))$.
$P$ is normed because $P \trans[\tau] \nu z(\nullproc \parpi z(x).!\ov{a}b \parpi \nullproc) \stops$
but there is no a unique parallel decomposition of $P$ because $P \wbisim P \parpi P$.

We study processes  without stuttering transitions in Section \ref{subsec:pwss}.
Using the results developed in that section and Theorem~\ref{th:unique_decomposition_strong_case}, 
in Section~\ref{subsec:upd_for_wb} we prove that for finite processes
there is a unique parallel decomposition w.r.t. weak bisimilarity.

\subsection{Processes without stuttering steps}
\label{subsec:pwss}

For $\omega = a_1a_2\cdots a_n \in A_\tau^*$ with $n>0$, 
we write $P \transs[\omega] P'$ if there are processes
$P_0, P_1, \ldots, P_n$ s.t.  $P = P_0 \transs[a_1] P_1 \transs[a_2] \ldots \transs[a_n] P_n = P'$. 
If $\omega = \varepsilon$, then $P \transs[\omega] P'$ implies $P \transs P'$.

\begin{definition}
A process $P \in \finitePi$ is a process \emph{without stuttering transitions} if 
there are no $\omega \in A^*$ and $P', P'' \in \finitePi$ s.t. 
$P \transs[\omega] P' \trans[\tau] P''$ and $P' \wbisim P''$.
We denote with $\wossPi$ the set of processes of $\finitePi$ without stuttering transitions.
\end{definition}
In Section~\ref{subsec:depth} we discussed why we do not consider infinite processes, this discussion
also applies for weak bisimilarity.
By definition, $\wossPi \subseteq \finitePi$. 
This fact and Lemma~\ref{lemma:wossPi_model_finitePi} ensure that we can use 
processes in $\wossPi$ to define properties over equivalence classes of processes in $\finitePi$
w.r.t. weak bisimilarity.

To prove Lemma~\ref{lemma:wossPi_model_finitePi} we need to introduce some notation and Lemma~\ref{lemma:hnf}.
We write $\ov{x}(z).P$ to denote $\nu z.\ov{x}z.P$. We call $\ov{x}(z)$ a \emph{bound-output prefix}.
We use $\lambda$ to range over prefixes, including bound-outputs. 

\begin{lemma}
 \label{lemma:hnf}
 For all $P \in \finitePi$ there are prefixes $\lambda_1, \ldots, \lambda_n$ and 
 processes $P_1, \ldots, P_n$ such that $P \bisim \sum_{i=1}^n \lambda_i.P_i$.
\end{lemma}
\begin{proof}
 The proof proceeds by structural induction on $P$.
 Cases $P=\nullproc$, $P = \pi.P'$ are straightforward by definition.
 For the case $P = P_1 + P_2$, by induction hypothesis there are processes 
 $Q_1 = \sum_{i \in I} \lambda_i.P_i$ and $Q_2 = \sum_{j \in J} \lambda_j.P_j$ and bisimulations ${\relR_1}$
 and ${\relR_2}$ s.t. $P_k \relR_k Q_k$ for $k=1,2$.
 It is easy to see that $\{(P, Q_1+Q_2)\} \cup \relR_1 \cup \relR_2$ is a bisimulation.
 Case $P = P_1 \parpi P_2$ is straightforward by induction and 
 the Expansion Lemma for~${\bisim}$ in the $\pi$-calculus \cite[Lemma 2.2.14]{SW01}.
 Thanks this lemma we can state that for all $P = \sum_{i \in I} \lambda_i P_i$ and  
 $Q = \sum_{j \in J} \lambda_j Q_j$ there is $R = \sum_{k \in K} \lambda_k R_k$ s.t. 
 $P \parpi Q \bisim R$.
 Case $P = \nu z.P'$ proceeds by structural induction on $P'$. 
 If $P' = \nullproc$ then $\nu z.P' \wbisim \nullproc = \sum_{i=1}^0 \lambda_i.P_i$.
 If $P' = \pi.P''$ then there are three cases to analyse: 
 \begin{inparaenum}[(i)]
  \item $z \not \in n(\pi)$ then $P \bisim \pi.(\nu z.P'')$,
  \item $\pi = \ov{x}z$ then $\nu z.P'$ can be denoted by $\ov{x}(z)P''$,
  \item $z \in n(\pi)$ and $\pi \neq \ov{x}z$ then $\nu z.\pi.P'' \bisim \nullproc$.
 \end{inparaenum}
 If $P' = P_0 + P_1$ then notice that $\nu z(P_0 + P_1) \bisim \nu z.P_0 + \nu z.P_1$;
 by induction there are processes $\sum_{i \in I} \lambda_i P_i$ and
 $\sum_{j \in J} \lambda_j P_j$ s.t. $\nu z.P_0 \bisim \sum_{i \in I} \lambda_i P_i$ and 
 $\nu z.P_1 \bisim \sum_{j \in J} \lambda_j P_j$.
 From this point, the proof follows as in the case $P=P_1+P_2$. 
%  therefore 
%  $P = \nu z(P_0 + P_1) \wbisim \sum_{i \in I} \lambda_i P_i + \sum_{j \in J} \lambda_j P_j$ 
%  because $+$ preserves ${\wbisim}$.	
%  
 Finally, case $P' = P_0 \parpi P_1$ can be reduced to the previous case using
 the Expansion Lemma for~${\bisim}$ (\cite[Lemma 2.2.14]{SW01}).
\end{proof}

\begin{lemma}
 \label{lemma:wossPi_model_finitePi}
 For every process $P \in \finitePi$ there is $Q \in \wossPi$ s.t. $P \wbisim Q$.
\end{lemma}
\begin{proof}
 The proof of the result follows by complete induction on $n = \depth P$. 
 By Lemma~\ref{lemma:hnf} for $P \in \finitePi$ there are prefixes $\lambda_1, \ldots, \lambda_n$ and 
 processes $P_1, \ldots, P_n$ such that $P \bisim \sum_{i=1}^n \lambda_i.P_i$.
 By induction and Lemma~\ref{lemma:trans_reduces_depth}, for each $i$ there is $Q_i \in \wossPi$ s.t. $P_i \wbisim Q_i$.  
 Then if we define 
 $$
 \textstyle
 Q = \sum_{i \in \{1, \ldots n\} \text{ and } P \nwbisim P_i} \lambda_i.Q_i$$
 we get $Q$ s.t. $Q \in \wossPi$ and $P\wbisim Q$.
\end{proof}

We cannot restrict our attention only to processes in $\wossPi$ because the property
of not executing stuttering transitions is not preserved by parallel composition.
Consider the processes $P_0 = \nu z (\ov{a}z)$ and $P_1 = a(x).(\ov{x}b+\tau.\ov{c}b)$.
Both $P_0, P_1 \in \wossPi$ but $P_0 \parallel P_1 \not \in \wossPi$ because
$$P_0 \parallel P_1 =  \nu z (\ov{a}z) \parallel a(x).(\ov{x}b+\tau.\ov{c}b)
\trans[\tau]
\nu z ( \nullproc \parallel (\ov{z}b+\tau.\ov{c}b)) \bisim \tau.\ov{c}b \wbisim \ov{c}b
$$
If we compare this fact with the strong setting, we can say that it is not possible to prove a 
lemma similar to Lemma~\ref{lemma:finiteness_is_conservative} for processes in $\wossPi$.

As in Section~\ref{subsec:depth}, we conclude with a collection of theorems and lemmas.
Theorems~\ref{th:wbisim_implies_same_depth} and~\ref{th:parameters_depths_lesser_than_par_depth_weak_case} 
are equivalent, respectively, to  
Theorems~\ref{th:bisim_implies_same_depth} and~\ref{th:parameters_depths_lesser_than_par_depth_strong_case} 
but w.r.t. processes in $\wossPi$ and weak bisimilarity. 
Most of the lemmas are needed to prove these results and only a few of them are used in the next section.

\begin{lemma}
 \label{lemma:trans_over_wosspi}
 If $P \in \wossPi$ and $P \transs[\omega] Q$ for $\omega \in A_\tau^*$ then $Q \in \wossPi$.
\end{lemma}
\begin{proof}
 Suppose $Q \not \in \wossPi$, then there are $\omega' \in A^*$ and $Q', Q'' \in \finitePi$ s.t. 
 $Q \transs[\omega'] Q' \trans[\tau] Q''$ with $Q' \wbisim Q''$.
 Let $\tilde \omega$ obtained from $\omega$ by removing $\tau$'s actions. 
 Then $P \transs[\tilde\omega\omega'] Q' \trans[\tau] Q''$ and therefore $P \not \in \wossPi$, which is a contradiction.
\end{proof}

\begin{lemma}
 \label{lemma:mimick_is_forced_to_trans}
 If $P,Q \in \wossPi$ are s.t. $P \wbisim Q$ and $P \trans[\alpha] P'$ with $\alpha \in A_\tau$,
 then $Q \transs[\alpha] Q'$, i.e. $Q$ executes at least a transition, and $P' \wbisim Q'$ 
\end{lemma}
\begin{proof}
 If $\alpha \neq \tau$ the result is straightforward by Def.~\ref{def:wbisim}.
 If $\alpha = \tau$ and there is no transition $Q_0 \trans[\tau] Q_1$ s.t. 
 $Q \transs Q_0 \trans[\tau] Q_1 \transs Q'$ and $P' \wbisim Q'$ then 
 $P' \wbisim Q$ since $P \wbisim Q$.
 This implies that $P \wbisim Q \wbisim P'$, i.e. $P \trans[\alpha] P'$ is a stuttering
 transition. This contradicts $P\in \wossPi$.
\end{proof}

\begin{theorem}
 \label{th:wbisim_implies_same_depth}
 If $P, Q \in \wossPi$ and $P \wbisim Q$ then $\depth P = \depth Q$.
\end{theorem}
\begin{proof}
We proceed by complete induction over $n = \depth P$. 
If $\depth P = 0$, then $P \bisim \nullproc$ and moreover $P \wbisim \nullproc$. 
Because $P \wbisim Q$ and $P \bisim \nullproc$, there is no $\alpha \in A$ s.t. $Q \transs[\alpha]$.
Taking into account this fact, if there is $Q'$ s.t. $Q \trans[\tau] Q'$,  $Q'$ is such that $Q' \wbisim \nullproc$.  
This creates a contradiction because $Q \trans[\tau] Q'$ is a stuttering transition and 
${Q \in \wossPi}$. Then $Q \bisim \nullproc$ and therefore $\depth Q = 0 =\depth P$. 
Suppose $\depth P = n + 1$. 
Let ${\omega = \alpha\omega' \in A_\tau^*}$ and $P'$ be s.t. 
$\length(\omega) = n + 1$ and $P \trans[\alpha] P' \trans[\omega']$.
Because $P \wbisim Q$ and Lemma~\ref{lemma:mimick_is_forced_to_trans} there are $Q_0$, $Q_1$, $Q'$ s.t. 
$Q \transs Q_0 \trans[\alpha] Q_1 \transs Q'$ and $P' \wbisim Q'$.
By induction $\depth{P'} = \depth{Q'}$ and therefore 
$\depth{Q} \geq \depth{P} = n+1$. 
We prove now that when we assume $\depth{Q} > n+1$ we reach a contradiction;
it then follows that $\depth{Q} = n+1$. 
Assume $\depth{Q} > n+1$ and 
let $\omega = \alpha\omega' \in A_\tau^*$ be such that $Q \trans[\alpha] \tilde Q \trans[\omega'] \tilde Q' \stops$
and $\length(\omega) = \depth{Q}$.
Because $P \wbisim Q$ and Lemma~\ref{lemma:mimick_is_forced_to_trans} there is $\tilde P$ s.t. 
$P \transs[\alpha] \tilde P$, $\tilde P \wbisim \tilde Q$ and $\depth{P} = n+1 > \depth{\tilde P}$.
By the complete induction $\depth{\tilde P} = \depth{\tilde Q}$.
Then, we reach a contradiction, $n+1 > \depth{\tilde P} = \depth{\tilde Q} \geq n + 1$. 
\end{proof}

\begin{lemma}
 \label{lemma:visible_act_change_the_eq_classes}
 If $P \in \finitePi$ and $P \trans[\alpha] P'$ with $\alpha \neq \tau $ then $P \nwbisim P'$.
\end{lemma}
\begin{proof}
 Let $\omega \in A^*$ be the largest sequence s.t. $P' \transs[\omega] P'' \stops$. 
 Then there is no $Q$ s.t. $P' \transs[\alpha\omega] Q$.
 On the other hand $P \transs[\alpha\omega]$, therefore $P \nwbisim P'$.
\end{proof}

\begin{lemma}
 \label{lemma:depth_one_implies_one_visible_action}
 If $P \in \wossPi$ is s.t. $P \nbisim \nullproc$ and 
 for all $P' \in \finitePi, \alpha \in A_\tau$, $P \trans[\alpha] P' \stops$, then 
 there is $\alpha' \neq \tau$ s.t. $P \trans[\alpha']$.
\end{lemma}
\begin{proof}
 Because $ P \nbisim \nullproc$ there is $\alpha$ s.t. $P \trans[\alpha]$.
 If for all $P' \in \finitePi, \alpha \in A_\tau$, $P \trans[\alpha] P' \stops$ 
 and $\alpha = \tau$ then $P \wbisim \nullproc$ and therefore all transitions that can be executed 
 by $P$ are stuttering transitions. This contradicts $P \in \wossPi$.
\end{proof}

\begin{theorem}
%  \label{lemma:parameters_depths_lesser_than_par_depth_weak_case}
 \label{th:parameters_depths_lesser_than_par_depth_weak_case}
 If $P, Q, R \in \wossPi$, $P \nwbisim \nullproc$, $Q \nwbisim \nullproc$ and 
 $P \parallel Q \wbisim R$ then $\depth P < \depth R$ and $\depth Q < \depth R$.
\end{theorem}
\begin{proof}
 We prove $\depth P < \depth R$, the proof that $\depth{Q}<\depth{R}$ is analogous. 
 Note that, since $Q \nwbisim \nullproc$, there is $Q'$ with $\depth{Q'}=1$ that is reachable
 from $Q$, that is, there exists $\omega \in A^*$ s.t.
%  Because $Q \nwbisim \nullproc$ there are
%  $Q'$ and 
 $Q \transs[\omega] Q'$ and $Q' \nbisim \nullproc$ (we remark the symbol $\nbisim$) and 
 for all $Q'' \in \finitePi, \alpha \in A_\tau$, $Q' \trans[\alpha] Q'' \stops$.
 Then, by Lemma~\ref{lemma:trans_over_wosspi}, $Q' \in \wossPi$
 and, by Lemma~\ref{lemma:depth_one_implies_one_visible_action}, $Q' \trans[\alpha]$ with $\alpha \neq \tau$. 
 Furthermore, by Convention~\ref{conv:names2} and the symmetric version of rule (\ref{trans:par-l})
 we have that $P \parallel Q \transs[\omega] P \parallel Q'$.  
 By Lemma~\ref{lemma:visible_act_change_the_eq_classes}, $P \parallel Q' \trans[\alpha] P \parallel \nullproc$
 and $\alpha \neq \tau$
 imply $P \parallel Q' \nwbisim P \parallel \nullproc \wbisim P$.
 Because $P \in \wossPi$,
 whenever $S \in \wossPi$ and $S \wbisim P \parallel Q'$, $\depth S \geq 1 + \depth P$ ($\star$).
 Given that $R \wbisim P \parallel Q$, $P \parallel Q \transs[\omega] P \parallel Q'$ 
 implies there is $R'$ s.t. $R\transs[\omega]R'$ and $R' \wbisim P \parallel Q'$.
 By Lemma~\ref{lemma:trans_over_wosspi},  $R' \in \wossPi$.
 In addition, by ($\star$), $\depth{R'} \geq 1 + \depth P$.
 Finally $\depth{R} \geq \depth{R'} \geq 1 + \depth{P} > \depth{P}$. 
\end{proof}

\begin{lemma}
 \label{lemma:wbisim_weak_challenge}
 If $P, Q \in \wossPi$, $P \wbisim Q$ and $P \transs[\alpha] P'$, with $\alpha \in A_\tau$,
 then there is $Q'$ s.t. $Q \transs[\alpha] Q'$.
\end{lemma}
% \begin{proof}
%  It is well-known that the challenge transition in the definition of weak bisimilarity
%  can be changed by a weak challenge transition. 
% %  
%  This fact and Lemma~\ref{lemma:mimick_is_forced_to_trans} allow to conclude the result.
% \end{proof}

\subsection{Unique parallel decomposition}
\label{subsec:upd_for_wb}

The development in this section is similar to the development in Section~\ref{subsec:upd}, 
for this reason in some cases we use the same notation. 
This will not be a problem because both developments are independent. 
In order to use Theorem~\ref{th:unique_decomposition_strong_case} we need to define a commutative monoid 
with a decomposition order. The commutative monoid is defined by 
\begin{itemize}
 \item $\wpp = \{[P]_{\wbisim}: P \in \finitePi\}$ where $[P]_{\wbisim} = \{P' : P' \wbisim P\}$
 \item $e = \wec{\nullproc} \in \wpp$ .
 \item ${\parallel} \subseteq \wpp \times \wpp$ is s.t. $\wec{P} \parallel \wec{Q} = \wec{P \parpi Q}$
\end{itemize}
Notice that we cannot ensure that for all $P', P'' \in \wec{P}$, $\depth{P'} = \depth{P''}$. 
Then, we extend the notion of depth in the following way.
Define $\wecRest{P} = \wec{P} \cap \wossPi$. 
For $\wec{P} \in \wpp$, $\depth{\wec{P}} = \depth{P'}$ with $P' \in \wecRest{P}$.
This definition is sound because of 
Lemma~\ref{lemma:wossPi_model_finitePi} and Theorem~\ref{th:wbisim_implies_same_depth}.

\begin{lemma}
$\wpp$ with neutral element $\wec{\nullproc}$ and binary operation ${\parallel}$
is a commutative monoid.
I.e., ${\parallel} \subseteq \wpp \times \wpp$ satisfies the associativity, commutativity and identity properties.
\end{lemma}

We shall define the partial order ${\porder}$ over $\wpp$ using the relation
${\trans} \subseteq \wpp \times \wpp$ defined as follows:
\begin{align*}
 {\trans_0} 
 & = 
 \{(\wec{P}, \wec{Q}) :\exists P' \in \wecRest{P}, Q' \in \wecRest{Q} : 
 P' \transs[\alpha] Q', \alpha \in A_\tau 
 \\  &\qquad \qquad \qquad \quad
 \text{ and} \not \exists P_0, P_1 \in \wossPi \text{ s.t. } P_0 \nwbisim \nullproc,  P_1 \nwbisim \nullproc, P_0 \parpi P_1 \wbisim P \} \\[1ex]
 {\trans_{k+1}} 
 & = \{(\wec{P_0 \parpi P_1},\wec{Q_0 \parpi P_1} ): \wec{P_0} \trans_k \wec{Q_0}, P_1 \in \finitePi\} \\
 & \cup \{(\wec{P_0 \parpi P_1},\wec{P_0 \parpi Q_1}) : \wec{P_1} \trans_k \wec{Q_1}, P_0 \in \finitePi\} \\[1ex]
 {\trans} & = \bigcup_{k = 0}^{\infty} \trans_{k}
\end{align*}
The partial order ${\porder}$ is defined as the inverse of the reflexive-transitive closure of ${\trans}$  
i.e., ${\porder} = ({\trans}^*)^{-1}$. We write $\wec{P} \prec \wec{Q}$ if
$\wec{P} \porder \wec{Q}$ and $\wec{P} \neq \wec{Q}$.

Notice that in this case ${\trans}$ takes into account processes in $\wossPi$ and 
weak transitions that execute at least one transition.
Also notice that we are avoiding communications between the arguments of the parallel composition
in order to avoid scope extrusion.

Similarly to the strong setting, we need two lemmas, 
Lemmas~\ref{lemma:trans_are_real_trans_wec} and~\ref{lemma:trans_reduces_depth_wec},
to prove that ${\porder}$ is a partial order (Lemma~\ref{lemma:partial_order_weak_case}).
In addition, to prove that $\porder$ is a decomposition order, we need the Lemma~\ref{lemma:there_is_always_a_path_weak_case}
that is equivalent to Lemma \ref{lemma:there_is_always_a_path}.
The proofs of these results follow similarly to their respective counterpart in the strong setting. 
(The complete proofs are in the appendix.)

\begin{lemma}
 \label{lemma:trans_are_real_trans_wec}
 If $\wec{P} \trans \wec{Q}$ then for all $\tilde P \in \wec P$ there are $\alpha \in A_\tau$ and $\tilde Q \in \wec Q$
 s.t. $\tilde P \transs[\alpha] \tilde Q$.
\end{lemma}
% \begin{proof}
%  The proof proceeds by complete induction on $k$ w.r.t $\trans_k$.
% %  
%  If $\wec{P} \trans_0 \wec{Q}$ then there are $P' \in \wecRest P, Q' \in \wecRest Q$, $\alpha \in A_\tau$ s.t. 
%  $P' \transs[\alpha] Q'$. 
% %  
%  $\tilde P \in \sec{P}$ implies $\tilde P \bisim P'$, then there is $\tilde Q$ s.t. $\tilde P \transs[\alpha] \tilde Q$
%  and $\tilde Q \wbisim Q' \wbisim Q$, i.e. $\tilde Q \in \wec Q$.
% %  
%  We prove the inductive case. 
% %  
%  Suppose, w.l.o.g, $\wec{P} = \wec{P_0 \parpi P_1} \trans_{k+1} \wec{Q_0 \parpi P_1}$ and $\wec{P_0} \trans_{k} \wec{Q_0}$. 
% % 
%  By induction, for all $\tilde P_0 \in \wec{P_0}$ there are $\tilde Q_0 \in \wec{Q_0}$ and $\alpha \in A_\tau$ s.t. 
%  $\tilde P_0 \transs[\alpha] \tilde Q_0$. 
% %  
%  In addition, by {rule~\ref{trans:par-l}}, ${\tilde P_0 \parpi P_1} \transs[\alpha] \tilde Q_0 \parpi P_1$ and because 
%  $\wbisim$ is a congruence for the parallel operator, $\tilde P_0 \parpi P_1 \in \sec{P_0 \parpi P_1}$
%  and $\tilde Q_0 \parpi P_1 \in \wec{Q_0 \parpi P_1}$.
% %
%  Now, for arbitrary $\tilde P \in \sec{P_0 \parpi Q}$, we have $\tilde P \wbisim \tilde P_0 \parpi P_1$, therefore 
%  $\tilde P \transs[\alpha] \tilde Q$ for some $\tilde{Q} \wbisim \tilde Q_0 \parpi P_1$, 
%  and therefore $\tilde{Q} \in \wec{Q_0 \parpi P_1}$. 
% \end{proof}

\begin{lemma}
 \label{lemma:trans_reduces_depth_wec}
 If $\wec{P} \trans \wec{Q}$ then $\depth{\wec{P}} > \depth{\wec{Q}}$.
\end{lemma}
% \begin{proof}
%  By Lemmas~\ref	{lemma:wossPi_model_finitePi} and~\ref{lemma:trans_are_real_trans_wec}, 
%  there are $\tilde P \in \wecRest{P}$ and $\tilde Q \in \wec{Q}$ s.t. 
%  $\tilde P \transs[\alpha] \tilde Q$ with $\alpha \in A_{\tau}$.
% %  
%  Moreover, by Lemma~\ref{lemma:trans_over_wosspi}, we have $\tilde Q \in \wecRest{Q}$.
% %  
%  By definition of $\transs[\alpha]$, there are $P', Q'$ s.t.  
%  $\tilde P \transs P' \trans[\alpha] Q' \transs \tilde Q$.
% %  
%  By Lemma~\ref{lemma:trans_reduces_depth},  
%  $\depth{\tilde P} \geq  \depth{P'} > \depth{Q'} \geq \depth{\tilde Q}$.
% %  
%  To conclude, notice  $\depth{\wec{P}} =  \depth{\tilde P}$ and $\depth{\wec{Q}} = \depth{\tilde Q}$.
% \end{proof}

\begin{lemma}
 \label{lemma:partial_order_weak_case}
 ${\porder}$ is a partial order. 
\end{lemma}
% \begin{proof}
% The proof of Lemma~\ref{lemma:partial_order_weak_case} follows as the proof of 
% Lemma~\ref{lemma:partial_order_strong_case} but to prove that the relation is antisymmetric we use
% Lemma~\ref{lemma:trans_reduces_depth_wec}.
% \end{proof}

% To prove that ${\porder}$ is a decomposition order we need to introduce some auxiliary results. 
% 
% Lemmas~\ref{lemma:there_is_always_a_path_weak_case} and~\ref{lemma:transition_over_wec} 
% are respectively equivalent to 
% Lemmas~\ref{lemma:there_is_always_a_path} and~\ref{lemma:transition_over_equivc} in Section~\ref{subsec:upd}.

\begin{lemma}
 \label{lemma:there_is_always_a_path_weak_case}
 If $P \in \wossPi$ and $\depth P > 0$ then there is $Q$ s.t. $\wec{P} \trans \wec{Q}$.
\end{lemma}

We are ready to prove that $\porder \subseteq \wpp \times \wpp$ is a decomposition order. 
This proof does not present changes w.r.t. proof of Lemma~\ref{lemma:dec_order_strong_case}
except that for proving $\porder$ is Archimedean, we use Theorem~\ref{th:parameters_depths_lesser_than_par_depth_weak_case}.
Notice that there is no lemma equivalent to Lemma~\ref{lemma:sum_of_depths} in the weak setting.

\begin{lemma}
 \label{lemma:dec_order_weak_case}
 ${\porder} \subseteq \wpp \times \wpp$ is a decomposition order.
\end{lemma}

By Theorem~\ref{th:unique_decomposition_strong_case}, it follows that $\wpp$ has unique decomposition.
\begin{corollary}
 The commutative monoid $\wpp$ has unique decomposition.
\end{corollary}

\section{Final Remarks}
\label{sec:final_remarks}

In this paper we have proved that finite processes of the $\pi$-calculus satisfy
UPD w.r.t.\ both strong bisimilarity and
weak bisimilarity. 
We have obtained these results using the technique presented in
\cite{LO05} (see Theorem~\ref{th:unique_decomposition_strong_case}) and different properties that are satisfied in each setting.
For the strong setting, we had to prove properties related to the depth of processes.
For the weak setting, we had to prove properties related to processes that execute no stuttering transitions. 
Our results show that the abstract framework of \cite{LO05}
can be used in the context of the $\pi$-calculus, dealing with the
complications that arise from scope extrusion.
In addition, the same framework can be used to deal with the weak setting if one considers 
processes without stuttering transitions. %in the weak setting,
In this way, we have avoided the abstract technique introduced in \cite{Luttik16}
which is considerably more involved than the technique that we have
used in this paper.

\newtext{
In Section~\ref{sec:upd_strong_case} we showed with two examples that norm is not additive for the $\pi$-calculus 
and therefore some proofs in \cite{DELL16} are flawed. 
After pointing out this problem to Dreier et al., 
they proposed us an alternative definition of norm to solve it. 
Call this variant $\normt$.
Roughly, $\normt$ should not consider traces where there is a scope extrusion of processes.
We think this solution may work for the applied $\pi$-calculus, but is
not suitable for the variant of the $\pi$-calculus
considered in the present paper.
We first explain what is the problem in our context, and then why this problem is not present in applied $\pi$. 
In the first example in Section~\ref{sec:upd_strong_case},
we had $P = P_0 \parpi P_1 = \nu z(\ov{a}z) \parpi a(x).!\ov{x}a$.
Process $P$ is not normed, i.e. $\normt(P) = \infty$, 
because the only finite trace that the process executes, $P \trans[\tau] \nu z(\nullproc \parpi !\ov{z}a) \stops$,
goes through a scope extrusion. 
Now, consider the process 
$$P' = \nu z(\ov{a}z).(\nullproc \parpi a(x).!\ov{x}a) \ + 
     \ a(x).(\nu z(\ov{a}z) \parpi !\ov{x}a) \ + 
     \ \tau.(\nu z(\nullproc \parpi !\ov{z}a))\enskip;$$
it would be normed according to the alternative definition suggested
above (the $\tau$-transition from $P'$ is not the result of a scope extrusion).
Now, since $P'$ is just the expansion of $P$, it is clear that $P
\bisim P'$.  Thus, we find that the property of being normed is not
compatible with bisimilarity.  
Since the applied $\pi$-calculus does not include the construct for
non-deterministic choice needed for the expansion, this problem is not
present there.
% In the second example, $\normt(Q)=3$ because the execution $Q \trans[\tau] \nu z(\nullproc \parpi \ov{z}a)\stops$
% also goes through a scope extrusion. In addition, $\normt(Q) = 3 = \normt(Q_0) + \normt(Q_1)$ as we wanted.
% % 
% We still have to prove if this solution really works. 
% % 
% We remark that in order to implement it, the operational semantics of calculus has to modified: 
% a new internal action $\tau_c$ has be the result of the communication of two processes.  
% % 
% In case the solution works, we think that the results in Section~\ref{sec:upd_strong_case} 
% can be extended to take into account $\normt$.
% We leave this development as a future work.
}

% This paper leaves two open questions.
% % 
% First, we have not yet found a proof of UPD for normed processes in
% the strong setting, nor a counter example against the property.
% % 
% We recall that we cannot use the norm of a process to do inductive reasoning, so we use depth, and because of 
% this we restricted our development in the strong setting to finite processes. 

An open question that leaves this paper is related with UPD of the $\pi$-calculus w.r.t. \emph{strong full bisimilarity}\cite{SW01}.
Strong full bisimilarity is a stronger notion of bisimulation that is a congruence for all constructs 
of the $\pi$-calculus.
% 
%  \begin{definition}
%  $P$ and $Q$ are \emph{strong full bisimilar},
%  notation $P \bisim^c Q$, if $P\sigma \bisim  Q\sigma$ for all substitutions $\sigma$.
%  \end{definition}
%  
We have tried to apply the abstract technique in this setting so far without success. 
% 
% After defining a partial commutative monoid and the depth of a processes taking into account 
% taking into account the universal quantification in the definition of strong full bisimilarity,
% a
When we tried to repeat the result in Section~\ref{sec:upd_strong_case}, taking into account 
the universal quantification in the definition of strong full bisimilarity,
a problem arose when we wanted to prove that the order is a decomposition order.
Particularly, we were not able to prove that the order is strict compatible. 
Notice that this problem is not present in the \emph{asynchronous $\pi$-calculus}\cite{SW01},
a well-known fragment of the $\pi$-calculus, 
because (strong) bisimilarity and (strong) full bisimilarity coincide.

% \begin{itemize}
%  \item We are really working with $\pi$-calculus and not with value-passing CCS with some $\pi$-calculus flavour.
%  \item We cannot ensure that the UPD result does not hold for normed processes with infinte bahaviors.
%  In addition, We do not know how to prove it because we do not know how to apply inductive reasoning.
%  \item We do not know how to deal with full version of the semantics. 
%  \item Using general techniques is cool.
%  \item Avoiding stuttering transitions allowed us to use \cite{LO05} and avoid using weak decomposition orders.
% \end{itemize}

 \medskip
\noindent
\emph{Acknowledgement.} The authors thank Daniel Hirschkoff for discussions, comments and 
suggestions on various drafts of this paper, and anonymous reviewers for their thorough reviews and good suggestions.

% \vspace{-0.2cm}
\bibliographystyle{eptcs}
\bibliography{decomposition}

\end{document}